\documentclass[conference]{IEEEtran}
\usepackage{amsmath}
\usepackage{multicol}
\usepackage{mathtools}
\usepackage{amssymb}
\usepackage{amsmath}  
\usepackage{subeqnarray}

\usepackage{algorithm}
\usepackage{algorithmic}        

\usepackage{cite}

\usepackage{array}  
\usepackage{url}

\usepackage{amsthm}      
\theoremstyle{plain}     
\newtheorem{prop}{Proposition}

\theoremstyle{definition}

\theoremstyle{remark}

\usepackage{graphicx}
\usepackage{subfigure}

\begin{document}

%
\title{Optimal Scheduling across Heterogeneous Air Interfaces of LTE/WiFi Aggregation}
\author{\IEEEauthorblockN{Yu Gu\IEEEauthorrefmark{1},
Qimei Cui\IEEEauthorrefmark{1},
Wei Ni\IEEEauthorrefmark{2},
Ping Zhang\IEEEauthorrefmark{1},
Weihua Zhuang\IEEEauthorrefmark{3}
}
\IEEEauthorblockA{\IEEEauthorrefmark{1}National Engineering Laboratory for Mobile Network\\
Beijing University of Posts and Telecommunications, Beijing, 100876, China \\
Correspondence authors: (guyu@bupt.edu.cn, cuiqimei@bupt.edu.cn)}
\IEEEauthorblockA{\IEEEauthorrefmark{2}CSIRO, Sydney, N.S.W. 2122, Australia}
\IEEEauthorblockA{\IEEEauthorrefmark{3}Department of Electrical and Computer Engineering, University of Waterloo, Waterloo, ON, N2L 3G1, Canada.}
}

\maketitle
\pagestyle{plain}
\begin{abstract}
LTE/WiFi Aggregation (LWA) provides a promising approach to relieve data traffic congestion in licensed bands by leveraging unlicensed bands. Critical challenges arise from provisioning quality-of-service (QoS) through heterogenous interfaces of licensed and unlicensed bands. In this paper, we minimize the required licensed spectrum without degrading the QoS in the presence of multiple users. Specifically, the aggregated effective capacity of LWA is firstly derived by developing a new semi-Markov model. Multi-band resource allocation with the QoS guarantee between the licensed and unlicensed bands is formulated to minimize the licensed bandwidth, convexified by exploiting Block Coordinate Descent (BCD) and difference of two convex functions (DC) programming, and solved efficiently with a new iterative algorithm. Simulation results demonstrate significant performance gain of the proposed approach over heuristic alternatives.
\end{abstract}


%
\IEEEpeerreviewmaketitle

\section{Introduction}
With the prevalence of smartphones, mobile data have been continuously growing and are expected to increase astoundingly 1000-fold by 2020 \cite{LAA}. Due to the prominent spectrum crunch, both industry and academia resort to unlicensed spectra, e.g., 5GHz ISM band, to accommodate
the rapidly growing mobile traffic \cite{7047304}. Approaches, such as cellular-to-WiFi offloading, have been proposed. In 3GPP Release-13, the aggregation of long-term evolution (LTE) and WiFi, also known as ``LTE/WiFi Aggregation (LWA)'', has been specified to leverage both the licensed and unlicensed spectra for data communications \cite{1}.

A few critical challenges arise from the aggregation of LTE and WiFi, especially when quality-of-service (QoS) is considered \cite{1}. The first challenge is that it is generally difficult for the contention-based WiFi access to guarantee QoS. This increasingly deteriorates, as the number of WiFi nodes, including both WiFi access points (APs) and stations (STAs), increases and the transmission collisions between the nodes become increasingly intensive. An other critical challenge is to optimally schedule the transmission of a data stream across different air interfaces associated with distinctive QoS-guaranteeing properties, such as decentralized contention-based access of WiFi and centrally coordinated access of LTE, while provisioning undegraded QoS to the stream \cite{1}. The complexity increases for a large number of streams to be scheduled in an LTE cell or WiFi AP  \cite{2}. An effective measure to quantify the QoS-provisioning capabilities of different air interference is critical, but yet to be developed, to implement the scheduling.

There are a small number of studies on traffic offloading between the licensed and unlicensed bands. In our previous work \cite{2}, we propose a unified framework supporting mobile converged network to implement LWA with a general description, which tightly integrates LTE and WiFi at the Medium access control (MAC) layer. To support a guaranteed bit rate (GBR) or non-GBR bearer, a joint access grant and resource allocation is proposed based on the QoS class indicator. In \cite{7}, a joint allocation of sub-carriers and power in the licensed and unlicensed band is presented to minimize the system power consumption with the aid of Lyapunov optimization. However, QoS has not been explicitly taken into account, and the impact of contention-based WiFi access on the QoS is yet to be investigated.

There are various works on QoS of data streams delivered through a single air-interface. Effective capacity (EC) theory has been applied to quantify the QoS, or more specifically the delay of a stream in a statistical fashion. In \cite{6840975}, the EC is used to measure the quality of mobile video traffic, and radio resources are allocated to maximize the EC under the Karush-Kuhn-Tucker (KKT) conditions. In our previous work \cite{LAA}, a semi-Markov model was developed to characterize the EC of licensed-assisted access networks, but still focus on a single air interface.

Different from the existing studies, in this paper, we focus on decoupling a stream with QoS into multiple sub-streams with different QoS in adaptation to heterogeneous air interfaces, and investigate an optimal scheduling of LWA systems to minimize the required licensed spectrum without degrading the QoS of multiple users. First, we derive the aggregated EC of LWA under statistical QoS requirements. In particular, a closed-form expression of EC in the unlicensed band is derived based on a semi-Markov model. Then, multi-band resource allocation problem with the QoS guarantee between licensed and unlicensed band is formulated to minimize the licensed bandwidth of LWA. The problem is a mixed-integer nonlinear and non-convex problem. A new iterative algorithm is proposed to convexify the problem as a series of subproblems based on Block Coordinate Descent (BCD) and difference of two convex functions (DC) programming. Simulation results demonstrate the significant performance gain over the heuristic alternatives.


\section{System Model}
We consider an LWA BS with a WiFi air interface in the unlicensed band 1, and a LTE air interface in the licensed band 2. The bandwidth of band $m~(m=1,2)$ is $B_{m}$. There are $N$ LWA users associated with the LWA BS. We assume Rayleigh block flat-fading channels in both bands; in other words, the channels remain unchanged during a time frame $T$, and vary across different time frames. Let $\gamma_{n,m}$ be the signal-to-noise ratio (SNR) of user $n$ in band $m$. The LWA BS can transmit to each of the users through both the licensed LTE interface and the unlicensed WiFi interface. In the latter case, the LWA BS acts as a standard WiFi access point (AP) to simultaneously deliver packets to all the users, using OFDMA techniques.

Apart from the LWA BS, there are also $L$ WiFi nodes operating in the unlicensed band within the coverage of the LWA BS. In the unlicensed band, all the $(L+1)$ nodes operate Distributed Coordination Function (DCF) to access the channel. Whenever a node has a packet to (re)transmit, it starts to sense the unlicensed band for a predefined period, termed ``distributed inter-frame space (DIFS)'', and generates an integer backoff timer randomly and uniformly within a contention window (CW) $[0,W_k)$, where the subscript ``$k$'' indicates the $k$-th retransmission of a packet, $k=0,1,\cdots,K-1$.
The timer counts down by one per slot, if the unlicensed band is free; otherwise, it freezes until the unlicensed band is free for DIFS again. Once the backoff timer reaches zero, a retransmission of the node is triggered.
The CW is doubled if the retransmission fails, i.e., no acknowledgment (ACK) is returned. After $K$ unsuccessful retransmission, the packet is discarded and the CW is reset to $W_0$.

Different from the single air-interface system, here packets arriving at the LWA BS and destined for user $n$ can be scheduled to traverse both air interfaces in the licensed and unlicensed bands. Two separate transmit queues are designed for the two air interfaces, as shown in Fig. 1(a). A binary variable $x_{m,n}$ is used to denote the band selected for the packets: $x_{m,n}=1$ if the LWA BS selects band $m$ to transmit the packets to user $n$; $x_{m,n}=0$, otherwise. For $x_{m,n}=1$, the bandwidth allocated to user $n$ is denoted by $\beta_{m,n}$.

The QoS of packets for user $n$ can be characterized by a QoS exponent, $\theta_n$. The larger $\theta_n$ is, the more stringent QoS is required. We propose to decouple $\theta_n$ between the two air interfaces, if both air interfaces are selected for the packets. Let $\theta_{m,n}$ denote the QoS exponent for the transmit queue of user $n$ in band $m$. We want to determine $x_{m,n}$, $\beta_{m,n}$ and $\theta_{m,n}$ $(m=1,2$, $n=1,\cdots,N)$, such that the maximum number of packets can be delivered without compromising ${\theta _n},\forall n$.

Each user, $n$, can also have a requirement of minimum data rate, denoted by $R_n$; a delay bound, denoted by $D_{th}^n$; and the maximum delay-bound violation probability threshold, denoted by $P_{th}^n$. The EC, denoted by $C_{m,n}(\theta_{m,n} )$, can be defined to be the maximum consistent arrival rate at the input of the transmit queue for user $n$ in band $m$, as given by \cite{LAA}
\begin{equation}\label{12}
C_{m,n}(\theta_{m,n} ) = \mathop {\lim }\limits_{t \to \infty }  - \frac{1}{{\theta_{m,n} t}}\log (\mathbb{E}\{ {e^{ - \theta_{m,n} S_{m,n}(t)}}\} )
\end{equation}
where $S_{m,n}(t)$ is the total number of bits transmitted to user $n$ in band $m$ during the period $[0,t)$, and $\mathbb{E}(\cdot)$ takes expectation.

By the EC theory \cite{6840975}, the delay-bound violation probability can be approximated to
\begin{equation}\label{1112}
\Pr\{ D > D_{th}^n\}  \approx {e^{ - {\theta _{m,n}}{C_{m,n}}({\theta _{m,n}})D_{th}^n}},\forall n \in {\mathcal{N}},\forall m \in {\mathcal{M}}.
\end{equation}

The delay-bound violation probability of the original data stream needs to satisfy
\begin{equation}\label{eq11c}
\frac{{\sum\limits_{m \in \mathcal{M}} {{x_{m,n}}{e^{ - {\theta _{m,n}}{C_{m,n}}({\theta _{m,n}})D_{th}^n}}{C_{m,n}}} ({\theta _{m,n}})}}{{\sum\limits_{m \in \mathcal{M}} {{x_{m,n}}{C_{m,n}}({\theta _{m,n}})} }} \le P_{th}^n,\forall n \in \mathcal{N}
\end{equation}
where ${\sum\limits_{m \in \mathcal{M}} {x_{m,n}}{{e^{ - {\theta _{m,n}}{C_{m,n}}({\theta _{m,n}})D_{th}^n}}{C_{m,n}}} ({\theta _{m,n}})}$ gives the total number of packets delivered to user $n$ via both bands before the delay bound, and ${\sum\limits_{m \in \mathcal{M}} {{x_{m,n}}{C_{m,n}}({\theta _{m,n}})} }$ is the total number of packets delivered to the user.
%

\begin{figure}[t]
\centering
\includegraphics[width=3.6in]{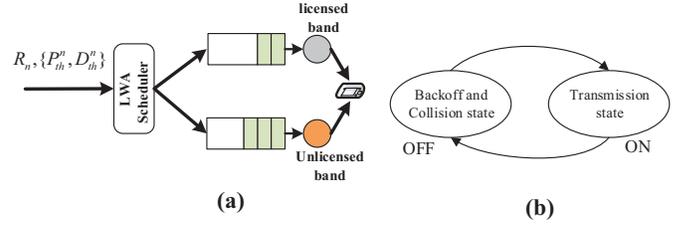}
\vspace{-0.4cm}
\caption{(a) Equivalent queuing model; (b) The On-off semi-Markov model for the unlicensed band.}
\label{fig.1}
\vspace{-0.2cm}
\end{figure}

\section{EC analysis of LWA}
In this section, we derive the EC of LWA through two heterogeneous air interfaces, which plays a key role to minimize the required licensed spectrum of the original stream while preserving QoS.
\subsection{EC in Unlicensed Band}
A semi-Markov model can be developed to evaluate the retransmission process of the LWA BS in the unlicensed band 1, as shown in Fig. \ref{fig.1}(b). For any user, $n$, the semi-Markov model consists of, namely, the ``on" state and the ``off" state. The ``on'' state corresponds to a successful retransmission of the user in the unlicensed band, with the average data rate ${\beta _{1,n}}{\log _2}(1 + {\bar \gamma _{1,n}})$ and the constant sojourn time ${t_s}$.

The ``off'' state accounts for the interval between two consecutive successful retransmissions of the LWA BS. It collects collided retransmissions between the two successful retransmissions, and the timeslots backed off in response to collisions. The transmission rate at the ``off'' state is zero, and the sojourn time is a random variable, denoted by ${t_{\rm off}}$.
Clearly, the transition probability matrix between the ``on'' and ``off'' states is ${\mathbf{P}} = \left[ {\begin{array}{*{20}{c}}
0&1\\
1&0
\end{array}} \right]$.


The probability generating function (PGF) of ${t_{\rm off}}$, denoted by $\hat{t}_{\rm off}$, can be evaluated as follows. First, the probability that a packet experiences $0<k \leq K-1$ retransmissions is given by
\begin{equation}\label{eq51}
{P_k} = \left\{ {\begin{array}{*{20}{l}}
{(1 - {p_c})p_c^k,  }~k \in \{ 0,1, \cdots ,K - 2\};\\
{p_c^{K - 1}{\rm{~~~~~~}},}~k = K - 1,
\end{array}} \right.
\end{equation}
where ${p_c}$ is the collision probability per slot. In response to the $k$ collisions, the total delay, denoted by $t_{\rm{off}}(k)$, can be written as
$t_{\rm off}(k) = k{t_c} + \sum\limits_{i = 1}^{{B_k}} {{X_i}},0 \le k \le K - 1,$
where ${t_c}$ is the duration of a collided retransmission in the unlicensed band; $X_i$ is the duration of the $i$-th timeslot; as the last successful retransmission, ${B_k}{\rm{ = }}\sum\limits_{j = 0}^k {{\eta _j}}$, is the total number of timeslots backed off in response to the $k$ collisions; and $\eta_j$ is the number of timeslots backed off in response to $j$-th collision.

Note that $X_i$ can be an idle minislot with duration of $\delta$, or a timeslot accommodating a collision-free transmission with duration of $t_s$, or a timeslot accommodating collided retransmissions with duration of $t_c$. $X_i$ is independent and identically distributed, and hence the subscript ``$i$'' is omitted in the following. The probability mass function (PMF) of $X$ can be given by
\begin{equation}\label{eq10}
{p_{X}}(x) = \left\{ {\begin{array}{*{20}{l}}
{\Pr \{ X = \sigma \}  = {{(1 - \tau )}^L} = 1 - {p_c}}\\
{\Pr \{ X = {t_s}\}  = L\tau {{(1 - \tau )}^{L - 1}}}\\
{\Pr \{ X = {t_c}\}  = {p_c} - L\tau {{(1 - \tau )}^{L - 1}}}
\end{array}} \right.
\end{equation}
where $\tau $ is the probability that a node transmits per slot. Both $\tau $ and ${p_c}$ can be uniquely determined \cite{14}. Then, the PGF of $X$, denoted by $\hat X(z)$, is given by
\begin{equation}\label{eq111}
\hat X(z) = \Pr\{ {X} = \sigma \} {z^\sigma } + \Pr\{ {X} = {t_s}\} {z^{{t_s}}} + \Pr\{ {X} = {t_c}\} {z^{{t_c}}}.
\end{equation}

Let ${\hat \eta _i}(z)$ denote the PGF of the number of time slots for the $i$-th retransmission, as given by
\begin{equation}\label{eq9}
{{\hat \eta }_i}(z) = \frac{1}{{{W_i}}}\frac{{1 - {z^{{W_i}}}}}{{1 - z}}\;\;,i = 0, \cdots ,K - 1.
\end{equation}
As a result, $\hat t_{\rm off}(z)$ is given by
\begin{equation}\label{eq122}
\begin{aligned}
\hat t_{\rm off}(z)= \sum\limits_{k = 0}^{K-1} {{P_k}\mathbb{E}[{z^{kt_{\rm off}(k)}}]} =\sum\limits_{k = 0}^{K - 1} {[{P_k}{z^{k{t_c}}}\prod\limits_{i = 0}^k {{{\hat \eta }_i}(\hat X(z))} ]}.
\end{aligned}
\end{equation}
The moment generating function (MGF) of $t_{\rm off}$, i.e., ${M_{\rm off}}(x) = \hat t_{\rm off}({e^x})$, can be achieved by substituting $z=e^x$ into (\ref{eq122}).

With reference to \cite{LAA}, we define two auxiliary variables, namely $s$ and $u$, and construct a diagonal matrix $\Gamma (s,u)$, in which the diagonal elements are the MGFs of the semi-Markov model, given by
\begin{equation}\label{eq30}
\begin{aligned}
\mathbf{\Gamma} (s,u) = \left[ {\begin{array}{*{20}{c}}
{\hat t_{\rm off}({e^{ - u}})}&0\\
0&{{e^{({\beta _{1,n}}{\log _2}(1 + {\bar \gamma _{1,n}})s - u){t_{s}}}}}
\end{array}} \right].
\end{aligned}
\end{equation}
We also construct
\begin{equation}
{\bf{H}}{(s,u)}= \mathbf{\Gamma} (s,u){\mathbf{P}}
=\left[ {\begin{array}{*{20}{c}}
{\rm{0}}&{{e^{({\beta _{1,n}}{{\log }_2}(1 + {{\bar \gamma }_{1,n}})s - u){t_s}}}}\\
{{{\hat t}_{\rm off}}({e^{ - u}})}&0
\end{array}} \right].
\end{equation}
Note that ${\bf{H}}{(s,u)}$ is non-negative irreducible matrix, since it cannot be restructured into an upper-triangular matrix by row/column operations. The spectral radius of ${\bf{H}}{(s,u)}$ is denoted by $\phi (s,u) = \rho (\mathbf{H}(s,u))$, where $\rho(\cdot)$ denotes the spectral radius of a matrix. According to \cite{LAA}, the EC of the semi-Markov model can be evaluated as $\frac{u}{s}$, when $\phi (s,u) = 1$ and ${\theta _{m,n}} =  - s$. As a result, the EC of user $n$ in the unlicensed band, $C_{1,n}$, can be evaluated by solving $\phi ( - \theta_{1,n} , - \theta_{1,n} {C_{1,n}}) = 1$.

As an eigenvalue of $\mathbf{H}( - \theta_{1,n} , - \theta_{1,n} {C_{1,n}})$, $\phi ( - \theta_{1,n} , - \theta_{1,n} {C_{1,n}})$ satisfies
\begin{equation}\label{square matrix}
\begin{aligned}
&{\left| {{\bf{H}}( - {\theta _{1,n}}, - {\theta _{1,n}}{C_{1,n}}) - \phi ( - {\theta _{1,n}}, - {\theta _{1,n}}{C_{1,n}}){\bf{I}}} \right| =}\\
&\left| {\begin{array}{*{20}{c}}
{ - \phi ( - {\theta _{1,n}}, - {\theta _{1,n}}{C_{1,n}})}&{{{\hat t}_{off}}({e^{{\theta _{1,n}}{C_{1,n}}}})}\\
{{e^{ - {\theta _{1,n}}({\beta _{1,n}}{{\log }_2}(1 + {{\bar \gamma }_{1,n}}) - {C_{1,n}}){t_s}}}}&{ - \phi ( - {\theta _{1,n}}, - {\theta _{1,n}}{C_{1,n}})}
\end{array}} \right|\\
&= \phi {( - {\theta _{1,n}}, - {\theta _{1,n}}{C_{1,n}})^2} - {e^{ - {\theta _{1,n}}({\beta _{1,n}}{{\log }_2}(1 + {{\bar \gamma }_{1,n}}) - {C_{1,n}}){t_s}}}\\
&~~~{{\hat t}_{\rm off}}({e^{{\theta _{1,n}}{C_{1,n}}}}) = 0\\
\end{aligned}
\end{equation}
where $\bf{I}$ is the identity matrix.
Substitute $\phi ( - \theta_{1,n} , - \theta_{1,n} {C_{1,n}})=1$ into (\ref{square matrix}) and then take logarithms. The EC of user $n$ in the unlicensed band can be obtained by solving
\begin{equation} \label{gy5}
F({C _{1,n}}\theta_{1,n} ) - \beta_{1,n} \theta_{1,n} {\log _2}(1 + \bar \gamma_{1,n} ){t_s} = 0
\end{equation}
where $F(x) = \log (\hat t_{\rm off}({e^x})) + x{t_s}$ for notation simplicity.

As MGF $\hat t_{\rm off}({e^x})$ increases monotonically with $x$, $F(x)$ is monotonically increasing and therefore invertible. We define ${F^{ - 1}}(x)$ as the inverse function of $F(x)$. The closed-form expression for the EC of user $n$ in the unlicensed band 1, is given by
\begin{equation}\label{eq11}
{C_{1,n}}({\beta _{1,n}},{\theta _{1,n}}) = \frac{1}{{{\theta _{1,n}}}}{F^{ - 1}}({\beta _{1,n}}{\theta _{1,n}}{\log _2}(1 + {{\bar \gamma }_{1,n}}){t_s}).
\end{equation}
\subsection{EC in Licensed Band}
In the licensed band 2, the EC of user $n$ is given by\cite{11}
\begin{equation}\label{eq5}
\begin{array}{l}
{C _{2,n}}({\beta _{2,n}}, \theta _{2,n}) =  - \frac{1}{{{\theta _{2,n}}T}}\log ({\mathbb{E}_\gamma }\{ {e^{ - {\theta _{2,n}}{\beta _{2,n}}T{{\log }_2}(1 + {\gamma})}}\} )\\
=  - \frac{1}{{{\theta _{2,n}}T}}\log (\int_0^\infty  {{{(1 + \gamma )}^{ - \frac{{{\theta _{2,n}}{\beta _{2,n}}T}}{{\log 2}}}}} \frac{1}{{{{\bar \gamma }_{2,n}}}}{e^{ - \frac{\gamma }{{{{\bar \gamma }_{2,n}}}}}}d\gamma )\\
=  - \frac{1}{{\theta_{2,n} T}}\log (\frac{1}{{\bar \gamma_{2,n} }}{{\bar \gamma_{2,n} }^{ - (\frac{{\theta_{2,n} \beta_{2,n} T}}{{2\log  2}} - 1)}}{e^{\frac{1}{{2\bar \gamma_{2,n} }}}}\Gamma )\\
 = \frac{\beta_{2,n} }{{2\log  2}}\log (\bar \gamma_{2,n} ) - \frac{1}{{2\theta_{2,n} T\bar \gamma_{2,n} }} - \frac{1}{{\theta_{2,n} T}}\log \Gamma
\end{array}
\end{equation}
where $\Gamma  = {W_{\frac{{\theta_{2,n} \beta_{2,n} T}}{{2\log  2}},\frac{{\ln 2 - \theta_{2,n} \beta_{2,n} T}}{{2\log  2}}}}(\frac{1}{{\bar \gamma }})$, and $W_{(\cdot,\cdot)}(\cdot)$ represents the Whittaker function.

As a result, the aggregated EC of user $n$, i.e., $\sum\limits_{m \in {\cal M}} {{x_{m,n}}{C_{m,n}}({\beta _{m,n}},{\theta _{m,n}})}$, can be evaluated through (\ref{eq11}) and (\ref{eq5}), subject to the QoS of the user (3). To this end, the QoS of each user can be guaranteed by best leveraging the bandwidth allocated to each user in both air interfaces while guaranteeing the QoS per band.

\section{Problem formulation}

We aim to minimize the total allocation of the licensed bandwidth while guaranteeing  the QoS of all the users. This can be formulated as
\begin{align*} \label{P1}
\stepcounter{equation}
\textbf{{P1:}}&\mathop {{{{\rm{minimize}}}}}\limits_{\{\beta _{m,n}\},\{\theta _{m,n}\},\{x _{m,n}\}} \sum\limits_{n \in {\cal N}} {x_{2,n}}{\beta _{2,n}} \tag{\arabic{equation}a}
\\
s.t.&\sum\limits_{m \in {\cal M}} {x_{m,n}}{{C_{m,n}}} ({\beta _{m,n}},{\theta _{m,n}}) \ge {R_n},\forall n \in {\cal N};\tag{\arabic{equation}b}
\\
&\sum\limits_{n \in \mathcal N} {x_{1,n}}{{\beta _{1,n}} \le {B_{1}}};\tag{\arabic{equation}c}
\\
&\frac{{\sum\limits_{m \in M} {{x_{m,n}}{e^{ - {\theta _{m,n}}{C_{m,n}}D_{th}^n}}{C_{m,n}}} }}{{\sum\limits_{m \in M} {{x_{m,n}}{C_{m,n}}} }} \le P_{th}^n,\forall n \in N;\tag{\arabic{equation}d}
\\
&{\beta _{m,n}} \ge 0,{\theta _{m,n}} \ge 0, \forall n \in {\mathcal{N}},\forall m \in {\mathcal{M}};\tag{\arabic{equation}e}\\
&{x_{m,n}} = \{ 0,1\}, \forall n \in {\mathcal{N}},\forall m \in {\mathcal{M}}.\tag{\arabic{equation}f}\\
\end{align*}
Eq. (15b) satisfies the minimum data rate of user $n$; (15c) guarantees that the total bandwidth allocated in the unlicensed band does not exceed $B_1$; (15d) is set to meet the QoS of user $n$; (15e) and (15f) are generic constraints to specify the feasible region of the problem. Clearly, \textbf{P1} is a combinatorial, mixed integer program. Moreover, ${{e^{ - {\theta _{m,n}}{C_{m,n}}D_{th}^n}}}$ and ${{x_{m,n}}{C_{m,n}}}$ are coupled in a multiplicative way in (15d). As a consequence, $C_{m,n}(\beta_{m,n},\theta_{m,n})$ is not joint convex in $\beta_{m,n}$ and $\theta_{m,n}$.

With $x_{m,n}$ being binary, we can rewrite
\begin{equation}
{x_{m,n}}{C_{m,n}}({\beta _{m,n}},{\theta_{m,n}}) = {C_{m,n}}({x_{m,n}}{\beta _{m,n}},{\theta_{m,n}}).
\end{equation}
Let ${\tilde \beta _{m,n}}{\rm{ = }}{x_{m,n}}{\beta _{m,n}}$, where $0 \le {\tilde \beta _{m,n}} \le {x_{m,n}}\Lambda$ and $\Lambda  > 0$ is a predefined constant.

In the case that band $m$ is selected to transmit packets to user $n$, (15d) can be rewritten as
\begin{equation} \label{aaa1}
{e^{ - {\theta _{m,n}}{C_{m,n}}D_{th}^n}} \le P_{th}^n.
\end{equation}
In the case that both the licensed and unlicensed bands are selected to transmit packets to user $n$, (15d) can be relaxed by using Chebyshev's sum inequality, as given by
\begin{equation} \label{aaa2}
\sum\limits_{m \in {\cal M}} {{e^{ - {\theta _{m,n}}{C_{m,n}}D_{th}^n}}{C_{m,n}}}  \le \frac{{\sum\limits_{m \in {\cal M}} {{e^{ - {\theta _{m,n}}{C_{m,n}}D_{th}^n}}} \sum\limits_{m \in {\cal M}} {{C_{m,n}}} }}{{\left| {\cal M} \right|}}
\end{equation}
where $|\cdot|$ stands for cardinality.

Replacing the left-hand side (LHS) of (15d) with the right-hand side (RHS) of (18), we relax (15d) to
\begin{equation} \label{aaa3}
\frac{{\sum\limits_{m \in {\cal M}} {{e^{ - {\theta _{m,n}}{C_{m,n}}D_{th}^n}}} }}{{\left| {\cal M} \right|}}\sum\limits_{m \in {\cal M}} {{C_{m,n}}}  \le P_{th}^n\sum\limits_{m \in {\cal M}} {{C_{m,n}}}.
\end{equation}
By combining (\ref{aaa1}) and (\ref{aaa3}), (15d) can be reformulated as
\begin{equation} \label{aaa42}
\sum\limits_{m \in {\cal M}} {{e^{ - {\theta _{m,n}}{C_{m,n}}D_{th}^n}} - 1 + {x_{m,n}}}  \le P_{th}^n\sum\limits_{m \in {\cal M}} {{x_{m,n}}} .
\end{equation}

Define two auxiliary variables ${\delta _{m,n}} = {\tilde \beta _{m,n}}{\theta _{m,n}}$ and ${a_{m,n}} = \frac{1}{{{\theta _{m,n}}}}$. We have ${\tilde \beta _{m,n}} = {\delta _{m,n}}{a_{m,n}}$. The EC of user $n$ can be equivalently rewritten as
\begin{equation} \label{gy3}
\addtocounter{equation}{1}
{C_{1,n}}({\delta _{1,n}},{a_{1,n}}) = {a_{1,n}}{F^{ - 1}}({\delta _{1,n}}{\log _2}(1 + {{\bar \gamma }_{1,n}}){t_s}),
\end{equation}
\begin{equation} \label{gy4}
{C_{2,n}}({\delta _{2,n}},{a_{2,n}}) =  - \frac{{{a_{2,n}}}}{T}\log ({E_\gamma }\{ {e^{ - {\delta _{2,n}}T{{\log }_2}(1 + {\gamma})}}\} ).
\end{equation}
Thus, (\ref{aaa42}) can be rewritten as
\begin{equation}
\begin{array}{l}
{e^{ - {F^{ - 1}}({\delta _{1,n}}{{\log }_2}(1 + {{\bar \gamma }_{1,n}}){t_s})D_{th}^n}}{\rm{ + }}{e^{\frac{{\log ({E_\gamma }\{ {e^{ - {\delta _{2,n}}T{{\log }_2}(1 + \gamma )}}\} )}}{T}D_{th}^n}}\\
 - {\rm{2}} + \sum\limits_{m \in M} {{x_{m,n}}}  \le P_{th}^n\sum\limits_{m \in M} {{x_{m,n}}} ,\forall n \in N.
\end{array}
\end{equation}

We can relax the binary constraint, (15f), as the intersection of the following regions \cite{6816086}
\begin{equation}
0 \le {x_{m,n}} \le 1,\forall n \in {\mathcal{N}},\forall m \in {\mathcal{M}};
\end{equation}
\begin{equation} \label{aaa41}
\sum\limits_{m \in {\cal M}} {\sum\limits_{n \in \cal N} {({x_{m,n}} - ({x_{m,n}})^2)} }  \le 0.
\end{equation}

As a result, \textbf{P1} can be relaxed to a continuous optimization problem, given by
\begin{align*}
\stepcounter{equation}
\textbf{{P2:}}&\mathop {{\rm{minimize}}}\limits_{\{ {\delta _{m,n}}\} ,\{ {a_{m,n}}\} ,\{ {x_{m,n}}\} } \sum\limits_{n \in \mathcal{N}} {{\delta _{2,n}}{a_{2,n}}} \tag{\arabic{equation}a}
\\
s.t.&\sum\limits_{m \in M} {{C_{m,n}}} ({\delta _{m,n}},{a_{m,n}}) \ge {R_n},\forall n \in N;\tag{\arabic{equation}b}
\\
&\sum\limits_{n \in N} {{\delta _{1,n}}{a_{1,n}}}  \le {B_1};\tag{\arabic{equation}c}
\\
& \sum\limits_{m \in {\cal M}} {\sum\limits_{n \in \cal N} {({x_{m,n}} - ({x_{m,n}})^2)} }  \le 0; \tag{\arabic{equation}d}
\\
&\begin{array}{l}
{e^{ - {F^{ - 1}}({\delta _{1,n}}{{\log }_2}(1 + {{\bar \gamma }_{1,n}}){t_s})D_{th}^n}}{\rm{ + }}{e^{\frac{{\log ({E_\gamma }\{ {e^{ - {\delta _{2,n}}T{{\log }_2}(1 + \gamma )}}\} )}}{T}D_{th}^n}}\\
 - {\rm{2}} + \sum\limits_{m \in M} {{x_{m,n}}}  \le P_{th}^n\sum\limits_{m \in M} {{x_{m,n}}} ,\forall n \in \mathcal{N};
\end{array}\tag{\arabic{equation}e}
\\
& 0 \le {{\delta _{2,n}}{a_{2,n}}} \le {x_{m,n}}\Lambda,\forall n \in {\mathcal{N}},\forall m \in {\mathcal{M}};\tag{\arabic{equation}f}\\
&{\delta _{m,n}} \ge 0,{a _{m,n}} \ge 0, 0 \le {x_{m,n}} \le 1,\forall n \in {\mathcal{N}},\forall m \in {\mathcal{M}}.\tag{\arabic{equation}g}\\
\end{align*}
\vspace{-1cm}

The objective is still non-convex, since $\delta _{m,n}$ and $a_{m,n}$ are still coupled. Nevertheless, we can achieve a partial optimum\footnote{$(x^*,y^*)$ is called a partial optimum of $f$ on $B=X \times Y$, if $f(x^*,y^*) \leq f(x,y^*), \forall x \in X$, and $f(x^*,y^*) \leq f(x^*,y), \forall y \in Y$.} using the following propositions.

\begin{prop}
Given $\{\delta _{m,n}\}$ and $\{x_{m,n}\}$, P2 is linear programming in $\{a_{m,n}\}$.
\end{prop}
\begin{proof}
Since (27a), (27b), and (27c) are affine, the proposition can be proved.
\end{proof}
\begin{prop}
Given $\{a_{m,n}\}$, P2 can be reformulated by using difference of convex (DC) programming, as given by
\begin{equation} \label{aaaad6}
\begin{aligned}
\textbf{\emph{P3:}}\mathop {{\rm{minimize}}}\limits_{\{ {\delta _{m,n}}\} ,\{ {x_{m,n}}\} }  &\sum\limits_{n \in \mathcal N} {{\delta _{2,n}}{a_{2,n}}}+ \lambda \sum\limits_{m \in \mathcal M} {\sum\limits_{n \in \mathcal N} {{x_{m,n}}}} \\
&- \lambda \sum\limits_{m \in \mathcal M} {\sum\limits_{n \in \mathcal N} {{{({x_{m,n}})}^2}} }
\end{aligned}
\end{equation}
\begin{equation}
s.t. (27b)-(27c),(27e)-(27g)
\end{equation}
where $\lambda$ is a large penalty factor.
\end{prop}
\begin{proof}
Applying Lyaponuv inequality \cite{6840975}, we have
\begin{equation}\label{eq42}
\mathbb{E}{[{\left| {{e^{{t_{\rm off}}}}} \right|^{{x_1}}}]^\alpha }\mathbb{E}{[{\left| {{e^{{t_{\rm off}}}}} \right|^{{x_2}}}]^{1 - \alpha }} \ge \mathbb{E}[{\left| {{e^{{t_{\rm off}}}}} \right|^{\alpha {x_1} + (1 - \alpha ){x_2}}}]
\addtocounter{equation}{1}
\end{equation}
and in turn, we have
\begin{equation}\label{eq12}
\begin{array}{l}
\alpha F({x_1}) + (1 - \alpha )F({x_2})\\
 = \log (\mathbb{E}{[{e^{{x_1}{t_{\rm off}}}}]^\alpha }\mathbb{E}{[{e^{{x_2}{t_{\rm off}}}}]^{1 - \alpha }}) + (\alpha {x_1} + (1 - \alpha ){x_2}){t_s}\\
 \ge \log (\mathbb{E}[{e^{(\alpha {x_1} + (1 - \alpha ){x_2}){t_{\rm off}}}}]) + (\alpha {x_1} + (1 - \alpha ){x_2}){t_s}\\
 = F(\alpha {x_1}{\rm{ + }}(1 - \alpha ){x_2}).
\end{array}
\end{equation}
In (\ref{eq12}), $F(x)$ is convex in $x$. Since $F(x)$ also monotonically increases, we can prove that ${C_{1,n}}$ is concave in $\delta_{1,n}$.

We can prove that ${C_{2,n}}$ is concave in $\delta_{2,n}$ using Lyaponuv inequality. The detailed proof is given in \cite{6840975}, and is omitted here.

We proceed to define
\begin{align}
f({\delta _{m,n}},{x_{m,n}}) &= \sum\limits_{n \in N} {\delta _{2,n}}{a_{2,n}} + \lambda \sum\limits_{m \in M} {\sum\limits_{n \in N} {{x_{m,n}}} }\\
g({x_{m,n}}) &= \lambda \sum\limits_{m \in M} {\sum\limits_{n \in N} {{{({x_{m,n}})}^2}} }.
\end{align}

Given ${a _{m,n}}$, $f({\delta _{m,n}},{x_{m,n}})$  and $g({x_{m,n}})$ are convex. Thus (\ref{aaaad6}) is the difference of two convex functions. Given that all the constraints are convex, P3 is a DC programming in $\{\delta _{m,n}, x _{m,n}\}$. Referring to \cite{7529170}, it can be proved that, given $\{a_{m,n}\}$ and a sufficiently large value of $\lambda$, \textbf{P3} is equivalent to \textbf{P2}.
\end{proof}

Based on Propositions 1 and 2, \textbf{P2} can be efficiently solved recursively by exploiting a block coordinated descent (BCD) framework \cite{6463487}, as summarized in Algorithm 1. Algorithm 1 consists of two steps. In the first step, given $\{\delta _{m,n}\}$ and $\{x _{m,n}\}$, \textbf{P2} is linear programming in ${a_{m,n}}$, which can be solved efficiently by using standard convex optimization techniques, such as the interior-point method. In second step, given ${a_{m,n}}$, \textbf{P3} is DC programming in $\{\delta _{m,n}, x _{m,n}\}$.

At the $l$-th iteration of DC programming, we use the first order Taylor expansion for $g( \cdot )$ to
approximate the objective function \cite{7529170}, as given by
\begin{equation} \label{gyy}
\begin{array}{l}
f({\delta _{m,n}},{x_{m,n}}) - g({x_{m,n}})\\
 \approx f({\delta _{m,n}},{x_{m,n}}) - g(x_{m,n}^{(l)}) - \left\langle {\nabla g(x_{m,n}^{(l)}),({x_{m,n}} - x_{m,n}^{(l)})} \right\rangle
\end{array}
\end{equation}
where $\left\langle  \cdot  \right\rangle$ denotes inner product. As a result, \textbf{P3} can be convexified, as given by
\begin{equation} \label{aaaa6}
\begin{array}{*{20}{l}}
{\mathop {{\rm{minimize}}}\limits_{\{ {\delta _{m,n}}\} ,\{ {x_{m,n}}\} }  \sum\limits_{n \in N} {\delta _{2,n}}{a_{2,n}} + \lambda \sum\limits_{m \in M} {\sum\limits_{n \in N} {{x_{m,n}}} } }\\
{ - \lambda \sum\limits_{m \in \mathcal{M}} {\sum\limits_{n \in \mathcal{N}} {\left( {{{(x_{m,n}^{(l - 1)})}^2} - 2x_{m,n}^{(l - 1)}({x_{m,n}} - x_{m,n}^{(l - 1)})} \right)} } }
\end{array}
\end{equation}
\begin{equation}
s.t. (27b)-(27c),(27e)-(27g).
\end{equation}

\begin{algorithm}[t]
  \caption{Optimal Scheduling algorithm across Heterogeneous Air Interfaces of LWA}
  \label{OPTIMAL RESOURCE ALLOCATION ALGORITHM}
  \begin{algorithmic}[1]
  \REQUIRE ~~\\
Given N users with minimum rata requirements ${R_n}$, and QoS requirement $\{P_{th}^n, D_{th}^n\}$, and channel gain ${\bar \gamma _{m,n}}$, and the total unlicensed bandwidth ${B_1}$; \\
  \ENSURE ~~\\
  \STATE Find any feasible point satisfying the constraint (27b)-(27g).
  \STATE Initialize $\{x^{(0)}_{m,n}\}$
  \REPEAT
  \STATE Keep ${\delta _{m,n}}$ and ${x_{m,n}}$ fixed for all user $n$ and band $m$. Optimize \textbf{P2} with respect to ${a_{m,n}}$.
  \REPEAT
  \STATE Keep $a_{n}$ fixed for all user $n$ and band $m$. Optimize (36) with respect to $\{\delta _{m,n}\}$ and $\{x _{m,n}\}$.
  \STATE Update $\{ x_{m,n}^*\}  \to \{ x_{m,n}^{(l)}\}$.
  \UNTIL{convergence}
  \UNTIL{convergence}
  \end{algorithmic}
\end{algorithm}

Algorithm 1 is convergent. This is because ${\delta _{m,n}}$, ${x _{m,n}}$ and ${a_{m,n}}$ can be recursively updated in sequel to reduce the objective function of \textbf{P2}, which monotonically decreases. Also, as the QoS requirement is finite, ${\delta _{2,n}}{a_{2,n}}$ is lower bounded.

The first step, exploiting the interior-point method, has a complexity of ${\mathcal O}({(MN)^3})$. In the second step, there are totally $2MN$ variables and $(2N+MN+1)$ convex and linear constraints in (36). Thus, the complexity of DC programming is $\mathcal{O}({(2MN)^3}(2N + MN +1))$ \cite{6816086}.  As a result, the overall complexity of the proposed algorithm is $\mathcal O({(MN)^6}(2N+ MN+1))$.

\section{Performance Evaluation}
In this section, monte-carlo simulations are run to evaluate the proposed algorithm. Assume that the time frame length of LTE $T = 1{\kern 1pt} ms$. The channels of licensed and unlicensed bands follow the ITU-UMi Models. There are 8 users uniformly distributed in the coverage overlapping area of the LWA BS and $L$ WiFi nodes. All users are set up with an identical minimum data rate ${R_n} = 1~{\rm{Mbps}},\forall {n}$ and delay threshold $D^n_{th} = 0.2~\rm s$. For comparison purpose, we also simulate the following two heuristic schemes:

\textbf{Sequential allocation scheme} (SAS): The LWA BS sorts users in the descending order of SNR and sequentially allocates the unlicensed and licensed spectrum to users. First, the LWA BS sequentially allocates the unlicensed bandwidth to the ordered user, until the unlicensed bandwidth is used up or the QoS of all users are satisfied. If the unlicensed band is insufficient to meet all the service requirements, repeats this allocation in the licensed band.

\textbf{Static mapping scheme} (SMS): A static mapping table is maintained. Let $\gamma$ denote a preconfigured fraction ($0 \le \gamma  \le 1$) of the minimum required rates are allocated to the unlicensed band, which be decided according to the QoS Class Indicator (QCI) or the types of traffics. Without loss of generality, we set $\gamma  = 0.6$.

\begin{figure}[t]
\centering
\includegraphics[width=3in]{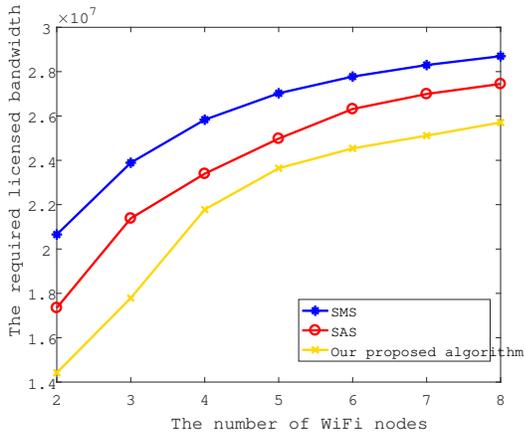}
\vspace{-0.4cm}
\caption{The required licensed bandwidth versus the number of WiFi APs.}
\label{fig.2}
\vspace{-0.2cm}
\end{figure}
In Fig. 2, we evaluate the required licensed bandwidth, with the number of WiFi nodes in the surrounding. We observe that the number of WiFi nodes has a strong impact on the requirement of the licensed bandwidth. The licensed bandwidth increases with the number of WiFi nodes, since the collisions in the unlicensed band aggravates. We also see that our proposed algorithm can reduce the allocated licensed bandwidth by up to 16.89\% than other schemes, when there is a small number of WiFi nodes.

\begin{figure}[t]
\centering
\includegraphics[width=3in]{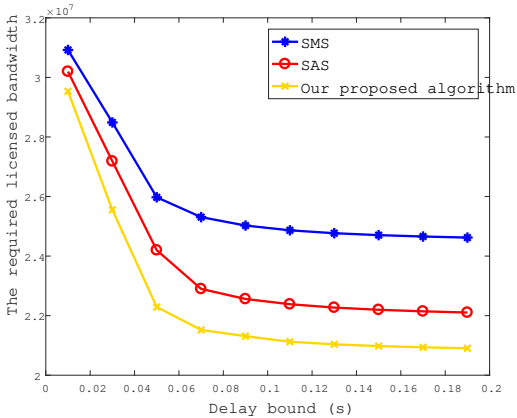}
\vspace{-0.4cm}
\caption{The licensed band versus the delay thresholds.}
\label{fig.3}
\vspace{-0.2cm}
\end{figure}
Fig. 3 shows the requirement of the licensed bandwidth versus the delay bound of the traffic. Within the coverage of the LWA BS, there are $N=4$ WiFi nodes. The figure shows that our proposed algorithm can substantially outperform the other schemes, and the gains of the proposed algorithm can be as high as up to 15.07\% and 5.38\%, as compared to SMS and SAS, respectively.

\emph{Remark}: The proposed algorithm is to minimize the total required licensed spectrum without degrading the QoS in the presence of multiple users. Another application is to maximize the total EC of LWA by optimizing multi-band resource allocation. By exploiting the aforementioned concavity of the EC again, as well as block coordinated descent framework, the EC maximization can be recursively solved by using our proposed algorithm.
\section{Conclusion}
In this paper, we analyze the aggregated EC of an LWA system based on a new Semi-Markov model. Then, we investigate an optimal scheduling of the LWA system to minimize the required licensed spectrum by a iterative algorithm. The data streams with QoS can be decoupled into multiple sub-streams with different QoS in adaptation to heterogeneous air interfaces. Simulation results show that the proposed algorithm can provide significant gains over the heuristic alternatives.

\section*{Acknowledgment}
The work was supported by National Nature Science Foundation of China Project (Grant No. 61471058 and 61461136002), Hong Kong, Macao and Taiwan Science and Technology Cooperation Projects (2016YFE0122900) and the 111 Project of China (B16006), BUPT Excellent Ph.D. Students Foundation (CX2017305) and in part by the Beijing Science and Technology Commission Foundation under Grant (201702005) and the China Scholarship Council (CSC).


\begin{thebibliography}{10}
\providecommand{\url}[1]{#1}
\csname url@samestyle\endcsname
\providecommand{\newblock}{\relax}
\providecommand{\bibinfo}[2]{#2}
\providecommand{\BIBentrySTDinterwordspacing}{\spaceskip=0pt\relax}
\providecommand{\BIBentryALTinterwordstretchfactor}{4}
\providecommand{\BIBentryALTinterwordspacing}{\spaceskip=\fontdimen2\font plus
\BIBentryALTinterwordstretchfactor\fontdimen3\font minus
  \fontdimen4\font\relax}
\providecommand{\BIBforeignlanguage}[2]{{%
\expandafter\ifx\csname l@#1\endcsname\relax
\typeout{** WARNING: IEEEtran.bst: No hyphenation pattern has been}%
\typeout{** loaded for the language `#1'. Using the pattern for}%
\typeout{** the default language instead.}%
\else
\language=\csname l@#1\endcsname
\fi
#2}}
\providecommand{\BIBdecl}{\relax}
\BIBdecl

\bibitem{LAA}
Q.~Cui, Y.~Gu, W.~Ni, and R.~P. Liu, ``Effective capacity of licensed-assisted
  access in unlicensed spectrum for 5G: From theory to application,''
  \emph{IEEE J. Sel. Areas Commun.}, vol.~35, no.~8, pp.
  1754--1767, Aug 2017.
\bibitem{6211487}
X. Tao and X. Xu and Q. Cui, ``An overview of cooperative communications,''
  \emph{IEEE Commun. Mag.}, vol.~50, no.~6, pp.
  65--71, Jun 2012.

\bibitem{7047304}
Q. Cui, Y. Gu, W. Ni, X. Zhang, X. Tao, P. Zhang, and R. P. Liu, ``Preserving Reliability to Heterogeneous Ultra-Dense Distributed Networks in Unlicensed Spectrum,''
 \emph{IEEE Commun. Mag.}, 2018, to appear,~\emph{http://arxiv.org/abs/1711.04614}.

\bibitem{1}
G.~RP-151022, ``LTE-WLAN radio level integration and interworking
  enhancement''.

\bibitem{2}
Q.~Cui, Y.~Shi, X.~Tao, P.~Zhang, R.~P. Liu, N.~Chen, J.~Hamalainen, and
  A.~Dowhuszko, ``A unified protocol stack solution for LTE and WLAN in future
  mobile converged networks,'' \emph{IEEE Wireless Commun.}, vol.~21,
  no.~6, pp. 24--33, December 2014.

\bibitem{7}
Y. Gu and Y. Wang and Q. Cui, ``A Stochastic Optimization Framework for Adaptive Spectrum Access and Power Allocation in Licensed-Assisted Access Networks,'' \emph{IEEE Access}, vol.~5, pp. 16484--16494, 2017.

\bibitem{6840975}
A.~A. Khalek, C.~Caramanis, and R.~W. Heath, ``Delay-constrained video
  transmission: Quality-driven resource allocation and scheduling,'' \emph{IEEE
  Journal of Selected Topics in Signal Processing}, vol.~9, no.~1, pp. 60--75,
  Feb 2015.

\bibitem{14}
G.~Bianchi, ``Performance analysis of the ieee 802.11 distributed coordination
  function,'' \emph{IEEE J. Sel. Areas Commun.}, vol.~18,
  no.~3, pp. 535--547, March 2000.

\bibitem{11}
F.~Jin, R.~Zhang, and L.~Hanzo, ``Resource allocation under delay-guarantee
  constraints for heterogeneous visible-light and rf femtocell,''  \emph{IEEE Trans. Wireless Commun.}, vol.~14, no.~2, pp. 1020--1034, Feb
  2015.

\bibitem{6816086}
E.~Che, H.~D. Tuan, and H.~H. Nguyen, ``Joint optimization of cooperative
  beamforming and relay assignment in multi-user wireless relay networks,''
   \emph{IEEE Trans. Wireless Commun.}, vol.~13, no.~10, pp.
  5481--5495, Oct 2014.

\bibitem{7529170}
B.~Khamidehi, A.~Rahmati, and M.~Sabbaghian, ``Joint sub-channel assignment and
  power allocation in heterogeneous networks: An efficient optimization
  method,'' \emph{IEEE Commun. Lett.}, vol.~20, no.~12, pp. 1089-7798, Dec, 2016.

\bibitem{6463487}
N.~Tekbiyik, T.~Girici, E.~Uysal-Biyikoglu, and K.~Leblebicioglu,
  ``Proportional fair resource allocation on an energy harvesting downlink,''
   \emph{IEEE Trans. Wireless Commun.}, vol.~12, no.~4, pp.
  1699--1711, April 2013.
\bibitem{6463487}
Q. Ye and W. Zhuang,
  ``Token-Based Adaptive MAC for a Two-Hop Internet-of-Things Enabled MANET,''
   \emph{IEEE Internet of Things Journal}, vol.~4, no.~5, pp.
  1739--1753, Oct 2017.

\end{thebibliography}

\end{document}